\documentclass[runningheads]{llncs}
\usepackage{graphicx}
\usepackage{amssymb,amsmath,makecell,multirow,todonotes,enumitem}
\usepackage{array}
\usepackage[algo2e,boxruled,linesnumbered,procnumbered]{algorithm2e}
\makeatletter
\renewcommand{\algocf@caption@boxruled}{%
  \hrule
  \hbox to \hsize{%
    \vrule\hskip-0.4pt
    \vbox{   
       \vskip\interspacetitleboxruled%
       \unhbox\algocf@capbox\hfill
       \vskip\interspacetitleboxruled
       }%
     \hskip-0.4pt\vrule%
   }\nointerlineskip%
}%
\makeatother
\usepackage[colorlinks=true,breaklinks=true,bookmarks=true,
     citecolor=blue,linkcolor=blue,bookmarksopen=false,draft=false]{hyperref}

\SetAlgorithmName{Procedure}{List of procedures}

\newcommand{\PreserveBackslash}[1]{\let\temp=\\#1\let\\=\temp}
\newcolumntype{C}[1]{>{\PreserveBackslash\centering}p{#1}}
\newcolumntype{R}[1]{>{\PreserveBackslash\raggedleft}p{#1}}
\newcolumntype{L}[1]{>{\PreserveBackslash\raggedright}p{#1}}

\newcommand{\cB}{\mathcal{B}}
\newcommand{\cC}{\mathcal{C}}

\newcommand{\cE}{\mathcal{E}}

\newcommand{\cH}{\mathcal{H}}

\newcommand{\cP}{\mathcal{P}}

\newcommand{\cS}{\mathcal{S}}

\newcommand\ChangeRT[1]{\noalign{\hrule height #1}}

\def\oa{\overline{a}}
\def\ob{\overline{b}}
\def\oc{\overline{c}}

\newcommand\price{\operatorname{{\it price}}}
\newcommand\cost{\operatorname{{\it cost}}}

\begin{document}

\title{Approximating minimum representations\\ of key Horn functions\thanks{Krist\'of is supported by the J\'anos Bolyai Research Fellowship of the Hungarian Academy of Sciences. This research was supported by SVV
project number 260 453.}}
%
%
\author{
Krist\'of B\'erczi\inst{1}\orcidID{0000-0003-0457-4573} \and
Endre Boros\inst{2}\orcidID{0000-0001-8206-3168} \and
Ond\v{r}ej \v{C}epek\inst{3}\orcidID{0000-0002-6325-0897} \and
Petr Ku\v{c}era\inst{3}\orcidID{0000-0002-7512-6260}\and
Kazuhisa Makino\inst{4}
}
\authorrunning{K. B\'erczi et al.}
%
\institute{MTA-ELTE Egerv\'ary Research Group, Department of Operations Research, E\"otv\"os University, Budapest, Hungary.
\email{berkri@cs.elte.hu}
\and
MSIS Department and RUTCOR, Rutgers University, New Jersey, USA.
\email{endre.boros@rutgers.edu}
\and
Charles University, Faculty of Mathematics and Physics, Department of Theoretical Computer Science and Mathematical Logic, Praha, Czech Republic.
\email{\{cepek,kucerap\}@ktiml.mff.cuni.cz}
\and
Research Institute for Mathematical Sciences (RIMS) Kyoto University, Kyoto, Japan.
\email{makino@kurims.kyoto.ac.jp}}

\maketitle              
\begin{abstract}
Horn functions form a subclass of Boolean functions and appear in many different areas of computer science and mathematics as a general tool to describe implications and dependencies. Finding minimum sized representations for such functions with respect to most commonly used measures is a computationally hard problem that remains hard even for the important subclass of key Horn functions. In this paper we provide logarithmic factor approximation algorithms for key Horn functions with respect to all measures studied in the literature for which the problem is known to be hard. 

\end{abstract}
\section{Introduction} \label{sec:intro}

A Boolean function of $n$ variables is a mapping from $\{0,1\}^n$ to $\{0,1\}$.
Boolean functions naturally appear in many areas of mathematics and computer
science and constitute a principal concept in complexity theory. In this paper we shall study an important problem connected to Boolean functions, a so called Boolean minimization problem, which aims at finding a shortest possible representation of a given Boolean function. The formal statement of the Boolean minimization problem (BM) of course depends on (i) how the input function is represented, (ii) how it is represented on the output, and (iii) the way how the output size is measured. 

One of the most common representations of Boolean functions are conjunctive normal forms (CNFs), the conjunctions of clauses which are elementary disjunctions of literals. There are two usual ways how to measure the size of a CNF: the number of clauses and the total number of literals (sum of clause lengths). It is easy to see that BM is NP-hard if both input and output is a CNF (for both above mentioned measures of the output size). This is an easy consequence of the fact that BM contains the CNF satisfiability problem (SAT) as its special case (an unsatisfiable formula can be trivially recognized from its shortest CNF representation). In fact, BM was shown to be in this case probably harder than SAT: while SAT is NP-complete (i.e. $\Sigma^p_1$-complete~\cite{C71}), BM is $\Sigma^p_2$-complete~\cite{U01} (see also the review paper~\cite{UVS06} for related results). It was also shown that BM is $\Sigma^p_2$-complete when considering Boolean functions represented by general formulas of constant depth as both the input and output for BM~\cite{BU11}.

Horn functions form a subclass of Boolean functions which plays a fundamental role in constructive logic and computational logic. They are important in automated theorem proving and relational databases. An important feature of Horn functions is that SAT is solvable for this class in linear time \cite{DOWLING1984267}. A CNF is Horn if every clause in it contains at most one positive literal, and it is pure Horn (or definite Horn in some literature) if every clause in it contains exactly one positive literal. A Boolean function is (pure) Horn, if it admits a (pure) Horn CNF representation. Pure Horn functions represent a very interesting concept which was studied in many areas of computer science and mathematics under several different names. The same concept appears as directed hypergraphs in graph theory and combinatorics, as implicational systems in artificial intelligence and database theory, and as lattices and closure systems in algebra and concept lattice analysis \cite{CASPARD2003241}. Consider a pure Horn CNF $\Phi=(\oa \vee b) \wedge (\ob \vee a) \wedge (\oa \vee \oc \vee d) \wedge (\oa \vee \oc \vee e)$ on variables $a,b,c,d,e$, where $\oa$ stands for the negation of $a$, etc. The equivalent directed hypergraph is $\cH=(V,\cE)$ with vertex set $V=\{a,b,c,d,e\}$ and directed hyperarcs
$\cE=\{ (\{a\}, b), (\{b\}, a), (\{a,c\}, d), (\{a,c\}, e)\}$. 
This latter can be expressed more concisely using a generalization of adjacency lists for ordinary digraphs in which all hyperarcs with the same body (also called source) are grouped together
$\{a\}: b, \{b\}: a, \{a,c\}: d,e$, or can be represented 
as an implicational (closure) system on variables $a,b,c,d,e$ defined by rules
$a \rightarrow b,  b \rightarrow a,  ac \rightarrow de$.

Interestingly, in each of these areas the problem similar to BM, i.e. a problem of finding the shortest equivalent representation of the input data (CNF, directed hypergraph, set of rules) was studied. For example, such a representation can be used to reduce the size of knowledge bases in expert systems, thus improving the performance of the system. The above examples show that a ``natural" way how to measure the size of the representation depends on the area. Six different measures and corresponding concepts of minimality were considered in~\cite{ADS86,CH11}: 
(B) number of bodies, 
(BA) body area,
(TA) total area,
(C) number of clauses,
(BC) number of bodies and clauses, and
(L) number of literals. For precise definitions, see Section~\ref{sec:prelim}.
With a slight abuse of notations we shall use (B), (BA), (TA), (C), (BC) and (L) to denote both the measures and the corresponding minimization problems. 


The only one of these six minimization problems for which a polynomial time procedure exists to derive a minimum representation is (B). The first such algorithm appeared in database theory literature~\cite{Maier80}. Different algorithms for the same task were then independently discovered in hypergraph theory~\cite{ADS86}, and in the theory of closure systems~\cite{GD86}.

For the remaining five measures it is NP-hard to find the shortest representation. There is an extensive literature on the intractability results in various contexts for these minimization problems~\cite{ADS86,HK93,Maier80}. It was shown that (C) and (L) stay NP-hard even when the inputs are limited to cubic (bodies of size at most two) pure Horn CNFs~\cite{BCK13}, and the same result extends to the remaining three measures. Note that if all bodies are of size one then the above problems become equivalent with the transitive reduction of directed graphs, which is tractable \cite{doi:10.1137/0201008}. It should be noted that there exists many other tractable subclasses, such as acyclic and quasi-acyclic pure Horn CNFs~\cite{HK95}, and CQ Horn CNFs~\cite{BCKK09}. There are also few heuristic minimization algorithms for pure Horn CNFs~\cite{BCK98}. 

It was shown that (C) and (L) are not only hard to solve exactly but even hard to approximate. More precisely, \cite{turan} shows that these problems are inapproximable within a factor $2^{\log^{1-\varepsilon}(n)}$ assuming $NP\subsetneq DTIME(n^{polylog(n)})$, where $n$ denotes the number of variables. In addition, \cite{BG14} shows that they are inapproximable within a factor $2^{\log^{1-o(1)}n}$ assuming $P\subsetneq NP$ even when the input is restricted to 3-CNFs with $O(n^{1+\varepsilon})$ clauses, for some small $\varepsilon>0$. It is not difficult to see that the same proof extends to (BC) and (TA) as well. On the positive side, (C), (BC), (BA), and (TA) admit $(n-1)$-approximations and (L) has an $\binom{n}{2}$-approximation \cite{HK93}. To the best of our knowledge, no better approximations are known even for pure Horn 3-CNFs.

Given a relational database, a key is a set of attributes with the property that a value assignment to this set uniquely determines the values of all other attributes~\cite{DM83,JU84}.
Analogously, we say that a pure Horn function is \emph{key Horn} if any of its bodies implies all other variables, that is, setting all variables in any of its bodies to one forces all other variables to one. This is a weaker concept than a database key, where setting the attributes in a key to any set of values determines the values of all remaining attributes. Key Horn functions are a generalization of a well studied class of \emph{hydra functions} considered in \cite{Sloan2017HydrasDH}. For this special class defined by the additional requirement that all bodies are of size two, a $2$-approximation algorithm for (C) was presented in \cite{Sloan2017HydrasDH} while the NP-hardness for (C) was proved in \cite{Kucera2014HydrasCO}. The latter result implies NP-hardness for hydra functions also for (BC), (TA), and (L). It is also easy to see that (B) and (BA) are trivial in this case.


In this paper we consider the minimization problems for key Horn functions. Any irredundant representation of a key Horn function has the same set of bodies, implying that problems (B) and (BA) are in P. We show that a simple algorithm gives a 2-approximation for (TA) and a $k$-approximation for (C), (BC), and (L), where $k$ is the size of a largest body. Our paper contains two main results. The first one improves the $(n-1)$-approximation bound for (C) and (BC) to $\min\{\lceil\log n\rceil+1,\lceil\log k\rceil +2\}$ in the case of key Horn functions. The second result improves the $\binom{n}{2}$-approximation bound for (L) to $\frac{108}{17}\lceil\log k\rceil+2$. Table~\ref{tab:results} summarizes the state of the art of Horn minimization and the results presented in this paper for key Horn functions. 

\begin{table}[htbp]
    \centering
    \caption{Complexity landscape of Horn and key Horn minimization, where the bold letters represent the results obtained in this paper. 
    Here $n$ and $k$ respectively denote the number of variables and the size of a largest body. All problems except those labeled by P are NP-hard. Inapproximability bounds for Horn minimization hold even when the size of the bodies are bounded by $k$ ($\geq 2$).}
    {\setlength{\extrarowheight}{4pt}
    \begin{tabular}{!{\vrule width 1pt}C{0.12\linewidth}
                    !{\vrule width 0.8pt}C{0.22\linewidth}
                    |C{0.12\linewidth}
                    !{\vrule width 0.8pt}C{0.43\linewidth}
                    !{\vrule width 1pt}}
        \ChangeRT{1pt}
        \multirow{2}{*}{Measure} & \multicolumn{2}{c!{\vrule width 0.8pt}}{ Horn} & Key Horn\\ [3pt] 
        \cline{2-4}
         & Inapprox. & Approx. & Approx.  \\[3pt]
        \ChangeRT{0.8pt}
        (B) & \multicolumn{2}{c!{\vrule width 0.8pt}}{P${}^{\mbox{\tiny\cite{Maier80}}}$} & P$^{\mbox{\tiny\cite{Maier80}}}$  \\[3pt]
        \ChangeRT{0.4pt}
        (BA) & $1^{\mbox{\tiny\cite{ADS86}}}$ & ${n-1}^{\mbox{\tiny\cite{HK93}}}$ & {\bf P}  \\[5pt]
        \ChangeRT{0.4pt}
        (TA) & ${2^{O(\log^{1-o(1)}n)}}^{\mbox{\tiny\cite{BG14}}}$ & ${n-1}^{\mbox{\tiny\cite{HK93}}}$ & {\bf 2}  \\[3pt]
        \ChangeRT{0.4pt}
        (C) & ${2^{O(\log^{1-o(1)}n)}}^{\mbox{\tiny\cite{BG14}}}$ & ${n-1}^{\mbox{\tiny\cite{HK93}}}$ &  $\boldsymbol{\min\{\lceil\log n\rceil + 1,\lceil\log k\rceil+2,k\}}$\\[3pt]
        \ChangeRT{0.4pt}    
        (BC) & ${2^{O(\log^{1-o(1)}n)}}^{\mbox{\tiny\cite{BG14}}}$ & ${n-1}^{\mbox{\tiny\cite{HK93}}}$ & $\boldsymbol{\min\{\lceil\log n\rceil + 1,\lceil\log k\rceil+2,k\}}$\\[3pt]
        \ChangeRT{0.4pt}
        (L) & ${2^{O(\log^{1-o(1)}n)}}^{\mbox{\tiny\cite{BG14}}}$ & $\binom{n}{2}^{\mbox{\tiny\cite{HK93}}}$  & $\boldsymbol{\min\{\frac{108}{17}\lceil\log k\rceil+2,k\}}$ \\[3pt]
        \ChangeRT{1pt}
    \end{tabular}}
    \label{tab:results}
\end{table}

The structure of our paper is as follows: Section \ref{sec:prelim} introduces the necessary definitions and notation, Section \ref{sec:lower} provides lower bounds for the measures we introduced, Section \ref{sec:approx} contains our results about approximation algorithms, while Section \ref{sec:ex} discusses the relation to the problem of finding a minimum weight strongly connected subgraph. 

\section{Preliminaries} \label{sec:prelim}

Let $V$ denote a set of variables. Members of $V$ are called \emph{positive} while their negations are called \emph{negative literals}. Throughout the paper, the number of variables is denoted by $n$. A \emph{Boolean function} is a mapping $f:\{0,1\}^V\rightarrow\{0,1\}$. The \emph{characteristic vector} of a set $Z$ is denoted by $\chi_Z$, that is, $\chi_Z(v)=1$ if $v\in Z$ and $0$ otherwise. We say that a set $Z\subseteq V$ is a \emph{true set} of $f$ if $f(\chi_Z)=1$, and a \emph{false set} otherwise. 


For a subset $\emptyset\neq B\subseteq V$ and $v\in V\setminus B$ we write $B\rightarrow v$ to denote the pure Horn clause $C=v\vee \bigvee_{u\in B}\overline{u}$. Here $B$ and $v$ are called the \emph{body} and \emph{head} of the clause, respectively. That is, a pure Horn CNF can be associated with a directed hypergraph where every clause $B\rightarrow v$ is considered to be a directed hyperarc oriented from $B$ to $v$. The \emph{set of bodies} appearing in a CNF representation $\Phi$ is denoted by $\mathcal{B}_{\Phi}$. We will also use the notation $B\rightarrow H$ to denote $\bigwedge_{v\in H}B\rightarrow v$. By grouping the clauses with the same body, a pure Horn CNF $\Phi=\bigwedge_{B\in\mathcal{B}_\Phi}\bigwedge_{v\in H(B)}B\rightarrow v$  can be represented as $\bigwedge_{B\in\mathcal{B}_\Phi} B\rightarrow H(B)$. The latter representation is in a one-to-one correspondence with the adjacency list representation of the corresponding directed hypergraph.  

For any pure Horn function $h$ the family of its true sets 
is closed under taking intersection and contains $V$. This implies that for any non-empty set $Z\subseteq V$ there exists a unique minimal true set containing $Z$. This set is called the \emph{closure} of $Z$ and is denoted by $F_h(Z)$. 
If $\Phi$ is a pure Horn CNF representation of $h$, then the closure $F_h(Z)$ can be computed in polynomial time by the following \emph{forward chaining procedure}. Set $F^0_{\Phi}(Z):=Z$. In a general step, if $F^i_{\Phi}(Z)$ is a true set then we set $F_\Phi(Z)=F^i_{\Phi}(Z)$. Otherwise, let $A\subseteq V$ denote the set of all variables $v$ for which there exists a clause $B\rightarrow v$ of $\Phi$ with $B\subseteq F^i_{\Phi}(Z)$ and $v\notin F^i_{\Phi}(Z)$, and set $F^{i+1}_{\Phi}(Z):=F^i_{\Phi}(Z)\cup A$. The result $F_\Phi(Z)$ does not depend on the particular choice of the representation $\Phi$, but only on the underlying function $h$, that is, $F_\Phi(Z)=F_h(Z)$.

A pure Horn function $h$ is \emph{key Horn} if it has a CNF representation of the form $\bigwedge_{B\in\mathcal{B}} B\rightarrow (V\setminus B)$ for some $\mathcal{B}\subseteq 2^V\setminus\{V\}$. We shall refer to $h$ as $h_{\mathcal{B}}$. Note that the same set of functions is defined if we restrict $\mathcal{B}$ to be \emph{Sperner}, that is, for any distinct $B,B'\in\mathcal{B}$ we have $B\not\subset B'$ and $B'\not\subset B$.

Assume now that $\Phi$ is a pure Horn CNF of the form $\bigwedge_{i=1}^m B_i\rightarrow H_i$ where $B_i\neq B_j$ for $i\neq j$. Note that the number of clauses in the CNF is $c_\Phi=\sum_{i=1}^m |H_i|$. The size of the formula can be measured in different ways:

\begin{itemize}
\item\textbf{(B) number of bodies}: $|\Phi|_B:=m$, 
\item\textbf{(BA) body area}:  $|\Phi|_{BA}:=\sum_{i=1}^m |B_i|$,
\item\textbf{(TA) total area}:  $|\Phi|_{TA}:=\sum_{i=1}^m (|B_i|+|H_i|)$, 
\item\textbf{(C) number of clauses (i.e., hyperarcs)}: $|\Phi|_C:=c_\Phi$, 
\item\textbf{(BC) number of bodies and clauses}: $|\Phi|_{BC}:=m+c_\Phi=\sum_{i=1}^m(|H_i|+1)$,
\item\textbf{(L) number of literals}:  $|\Phi|_L:=\sum_{i=1}^m \big((|B_i|+1)\cdot|H_i|\big)$.
\end{itemize}

These measures come up naturally in connection with directed hypergraphs, implicational systems, and CNF representations. The Horn minimization problem is to find a representation that is equivalent to a given Horn formula and has minimum size with respect to $|\cdot|_*$ where $*$ denotes one of the aforementioned functions. 

\section{Lower bounds} \label{sec:lower}

The present section provides some simple reductions of the problem and lower bounds for the size of an optimal solution. 

For a family $\mathcal{B}\subseteq 2^V\setminus\{V\}$, we denote by $\cal{B}^\bot$ the family of minimal elements of $\cal B$. Recall that $h_\mathcal{B}$ denotes the function defined by 
\begin{equation}
\Psi_\mathcal{B}=\bigwedge_{B\in\mathcal{B}} B\rightarrow (V\setminus B).  \label{eq:psi}  
\end{equation}

\begin{lemma}\label{lem:origbod}
For any measure ($*$) and for any $\mathcal{B}\subseteq 2^V\setminus\{V\}$, there exists a $|\cdot|_*$-minimum representation of $h_\mathcal{B}$ that uses exactly the bodies in $\mathcal{B}^\bot$. 
\end{lemma}
\begin{proof}
Take a $|\cdot|_*$-minimum representation $\Phi$ for which $|\mathcal{B}_\Phi\setminus\mathcal{B}^\bot|$ is as small as possible. First we show $\mathcal{B}_\Phi\subseteq \mathcal{B}^\bot$. Assume that $B\in \mathcal{B}_\Phi\setminus\mathcal{B}^\bot$. As $B$ is a false set of $h_\mathcal{B}$, there must be a clause $B'\rightarrow v$ in $\Psi_\mathcal{B}$ that is falsified by $\chi_B$, implying that $B'\subseteq B$. Therefore there exists a $B''\in\mathcal{B}^\bot$ such that $B''\subseteq B'\subseteq B$. If we substitute every clause $B\rightarrow v$ of $\Phi$ by $B''\rightarrow v$, then we get another representation of $h_\mathcal{B}$ since $B''\rightarrow v$ is a clause of $\Psi_\mathcal{B}$. Meanwhile, the $|\cdot|_*$ size of the representation does not increase while $|\mathcal{B}_\Phi\setminus\mathcal{B}^\bot|$ decreases, contradicting the choice of $\Phi$. 

Next we prove $\mathcal{B}_\Phi\supseteq \mathcal{B}^\bot$. If there exists a $B\in\mathcal{B}^\bot\setminus\mathcal{B}_{\Phi}$, then $B$ is a true set of $\Phi$ while it is a false set of $h_\cB$, contradicting the fact that $\Phi$ is a representation of $h_\cB$.\qed
\end{proof}

Lemma~\ref{lem:origbod} has two implications. It suffices to consider Sperner hypergraphs defining key Horn functions as an input, and more importantly, it is enough to consider CNFs using bodies from the input Sperner hypergraph when searching for minimum representations. For non-key Horn functions, this is not the case. 

From now on we assume that $\mathcal{B}$ is a Sperner family. We also assume that 
\[
\bigcup_{B\in\mathcal{B}} B=V  ~~~\text{  and  }~~~ \bigcap_{B\in\mathcal{B}} B=\emptyset.
\]
Indeed, if a variable $v\in V\setminus \bigcup_{B\in\mathcal{B}} B$ is not covered by the bodies, then there must be a clause with head $v$ and body in $\mathcal{B}$ in any minimum representation of $h_\mathcal{B}$, and actually one such clause suffices. Furthermore, if $v\in \bigcap_{B\in\mathcal{B}} B$, then we can reduce the problem by deleting it. None of these reductions affects the approximability of the problem.

Recall that the size of the ground set is denoted by $|V|=n$, while $|\mathcal{B}|=m$. The size of an optimal solution with respect to measure function $|\cdot|_*$ is denoted by $OPT_*(\mathcal{B})$. Using these notations
Lemma~\ref{lem:origbod} has the following easy corollary:

\begin{corollary}
We have $OPT_{B}(\mathcal{B})=m$ and $OPT_{BA}(\mathcal{B})=\sum_{B\in\mathcal{B}}|B|$. 
Therefore the minimization problems (B) and (BA) are solvable in polynomial time.\qed
\end{corollary}

For the remaining measures we prove the following simple lower bound.

\begin{lemma}\label{lem:lb}
$OPT_*(\mathcal{B})\geq m$ for all measures $*$, and $OPT_*(\mathcal{B})\geq n$ for $*\in\{TA,C,BC,L\}$.
Furthermore, $OPT_L(\cB)\geq \max\{n(\delta +1),2m\}$, where $\delta$ is the size of a smallest body in $\cB$. 
\end{lemma}
\begin{proof}
By definition, $|\cdot |_B$ is a lower bound for all the other measures, implying $OPT_*(\mathcal{B})\geq OPT_B(\cB)=m$.

To see the second part, observe that $|\cdot |_C$ is a lower bound for the three other measures. Therefore it suffices to prove $OPT_C(\cB)\geq n$. By the assumption that for every $v\in V$ there exists a $B\in\cB$ not containing $v$, we can conclude by the fact that the closure $F_{h_\cB}(B)=V$ and by the way the forward chaining procedure works that every CNF representation of $h_\cB$ must contain at least one clause with $v$ as its head. This implies $OPT_C(\mathcal{B})\geq n$.  

To see the last part note that every variable $v\in V$ is the head of at least one clause, the body of which is of at least size $\delta\geq 1$. Furthermore, since every body appears at least once and all clauses are of size at least $2$, the claim follows. 
\qed
\end{proof}

For a pair $S,T\subseteq V$ of sets, let $\price_*(S,T)$ denote the minimum $|\cdot|_*$-size of a CNF $\Phi$ for which $\cB_\Phi\subseteq \cB$ and $T\subseteq F_\Phi(S)$, that is,
\begin{equation}
\price_*(S,T)=\min_\Phi\big\{|\Phi|_*\mid\mathcal{B}_\Phi\subseteq\mathcal{B}, T\subseteq F_\Phi(S)\big\}. \label{eq:price}
\end{equation}

The following lemma plays a key role in our approximability proofs.

\begin{lemma}\label{lem:main}
Let $\mathcal{B}=\mathcal{B}_1\cup\dots\cup\mathcal{B}_q$ be a partition of $\mathcal{B}$ and let $B_i\in\mathcal{B}_i$ for $i=1,\ldots,q$. Then 
\begin{equation}
OPT_*(\mathcal{B})\geq \sum_{i=1}^q\min\{\price_*(B_i,B)\mid B\in\cB\setminus\cB_i\}
\end{equation}
for all six measures $*$.
\end{lemma}
\begin{proof}
Take a minimum representation $\Phi$ with respect to $|\cdot|_*$ which uses bodies only from $\mathcal{B}$. Such a representation exists by Lemma~\ref{lem:origbod}. We claim that the contribution of the clauses with bodies in $\mathcal{B}_i$ to the total size of $\Phi$ is at least $\min\{\price_*(B_i,B)\mid B\in\cB\setminus\cB_i\}$ for each $i=1,\ldots,q$. This would prove the lemma as the $\mathcal{B}_i$'s form a partition of $\mathcal{B}$.

To see the claim, take an index $i\in\{1,\ldots,q\}$ and let $B'$ be the first body (more precisely, one of the first bodies) not contained in $\mathcal{B}_i$ that is reached by the forward chaining procedure from $B_i$ with respect to $\Phi$. Every clause that is used to reach $B'$ from $B_i$ has its body in $\mathcal{B}_i$ and their contribution to the size of the representation is lower bounded by $\price_*(B_i,B')$, thus concluding the proof. \qed
\end{proof}

\section{Approximability results for (TA), (C), (BC), and (L)}\label{sec:approx}

Given a Sperner family $\mathcal{B}\subseteq 2^V\setminus\{V\}$, we can associate with it a complete directed graph $D_{\mathcal{B}}$ by defining $V(D_{\mathcal{B}})=\mathcal{B}$ and $E(D_{\mathcal{B}})=\mathcal{B}\times\mathcal{B}$. We refer to $D_{\mathcal{B}}$ as the \emph{body graph} of $\mathcal{B}$. 

For any subset $E'\subseteq E(D_{\mathcal{B}})$, define
\begin{equation}
\Phi_{E'}=\bigwedge_{(B,B')\in E'} B\rightarrow(B'\setminus B). \label{eq:sc}
\end{equation}
Note that if $E'\subseteq E(D_{\mathcal{B}})$ forms a strongly connected spanning subgraph of $D_{\mathcal{B}}$, then $\Phi_{E'}$ is a representation of $h_\mathcal{B}$. Let us add that not all representations arise this way, in particular, minimum representations might have significantly smaller size.

\begin{lemma} \label{lem:kapprox}
If $E'$ is a Hamiltonian cycle in $D_\cB$, then $\Phi_{E'}$ defined in \eqref{eq:sc} provides a $k$-ap\-prox\-i\-ma\-tion for all measures, where $k$ is an upper bound on the sizes of bodies in $\cB$.
\end{lemma}
\begin{proof}
By Lemma~\ref{lem:origbod}, there exists a minimum representation $\Phi$ of $h_\cB$ such that $\cB_\Phi=\cB$. Since $|B'\setminus B|$ is at most $k$ for all arcs $(B,B')\in E'$, the statement follows. \qed
\end{proof}

In fact, for (B) and (BA) \eqref{eq:sc} gives an optimal representation for any strongly connected spanning $E'$. Furthermore, if $E'$ is a Hamiltonian cycle, we get a $2$-approximation for (TA) based on the fact that the total area of any representation is lower bounded by $\sum_{B\in\mathcal{B}} |B|$.

\begin{theorem} \label{thm:tamin}
If $E'$ is a Hamiltonian cycle in $D_\cB$, then $\Phi_{E'}$ defined in \eqref{eq:sc} provides a $2$-ap\-prox\-i\-ma\-tion for (TA).
\end{theorem}
\begin{proof}
The size of $\Phi_{E'}$ is $|\Phi_{E'}|_{TA}=\sum_{i=1}^m (|B_i|+|B_{i+1}\setminus B_i|)\leq 2 \sum_{i=1}^m |B_i|\leq 2OPT_{TA}(\mathcal{B})$. \qed
\end{proof}

The observation that a strongly connected subgraph of the body graph corresponds to a representation of $h_\cB$, as in \eqref{eq:sc}, suggests the reduction of our problem to the problem of finding a minimum weight strongly connected spanning subgraph in a directed graph with arc-weight $\price_*(B,B')$ for $(B,B')\in E(D_\cB)$. The optimum solution to this problem (MWSCS) is an upper bound for the minimum $|\cdot|_*$-size of a representation of $h_\cB$. As there are efficient constant-factor approximations for MWSCS \cite{jaja}, this approach may look promising. There are however two difficulties: for measure (L), no polynomial time algorithm is known for computing $\price_L$; even when it is efficiently computable (for measures (C) and (BC)), the upper bound obtained in this way may be off by a factor of $\Omega(n)$ from the optimum (see Section~\ref{sec:ex} for a construction).

In what follows, we overcome these difficulties. For (C), instead of a strongly connected spanning subgraph, we compute a minimum weight spanning in-ar\-bores\-cence and extend that to a representation of $h_\cB$. The same approach works for (BC) as well.
For (L), the situation is more complicated. First, we develop an efficient approximation algorithm for $\price_L$.  
Next, we compute a minimum weight spanning in-ar\-bores\-cence where its root is pre-specified. Finally, we extend the corresponding CNF to a representation of $h_\cB$. 
We show that the cost of the arborescences built is at most a multiple of the optimum by a logarithmic factor, which in turn ensures the improved approximation factor.  

\subsection{Clause and body-clause minimum representations} \label{sec:emin}

In this section we consider (C) and (BC) and show that the simple algorithm described in Procedure 1 provides the stated approximation factor. We note that a minimum weight spanning in-arborescence of a directed graph can be found in polynomial time, see \cite{CL65,E67}.

\begin{algorithm2e}
\caption{Approximation of (C) and (BC)} \label{proc:base}

\SetAlgoLined

 Determine a minimum $\price_C$-weight spanning in-arborescence $T$ of $D_\cB$. \linebreak /$*$ Denote by $B_0$ the body corresponding to the root of $T$. $*$/
 
 Output $\Phi=\Phi_T\wedge B_0\rightarrow(V\setminus B_0)$. \linebreak /$*$ Here $\Phi_T$ is defined as in \eqref{eq:sc}. $*$/
\end{algorithm2e}

First we observe that $\price_C$ is easy to compute.

\begin{lemma} \label{lem:triangle}
$\price_C(B,B')=|B'\setminus B|$ for $B,B'\in\cB$.
\end{lemma}
\begin{proof}
Take a pure Horn CNF $\Phi$ attaining the minimum in \eqref{eq:price}. As every variable in $B'\setminus B$ is reached by the forward chaining procedure from $B$ with respect to $\Phi$, each such variable must be a head of at least one clause in $\Phi$. That is, $\Phi$ contains at least $|B'\setminus B|$ clauses. On the other hand, $B\rightarrow(B'\setminus B)$ uses exactly $|B'\setminus B|$ clauses, hence $\price_C(B,B')=|B'\setminus B|$ as stated. \qed
\end{proof}


\begin{lemma} \label{lem:minarb}
Let $T$ denote a minimum $\price_C$-weight spanning in-arborescence in $D_\cB$. Then $$|\Phi_T|_C\leq\lceil\log k\rceil OPT_C(\cB)+\max\{0,m-k\},$$ where $k$ is an upper bound on the sizes of bodies in $\cB$.
\end{lemma}
\begin{proof}
We construct a subgraph $T$ of $D_\cB$ such that (i) it is a spanning in-arborescence, and (ii) $|\Phi_T|_C\leq\lceil\log k\rceil OPT_C(\cB)+\max\{0,m-k\}$. We start with the digraph $T_1$ on node set $\cB$ that has no arcs. In a general step of the algorithm, $T_i$ will denote the graph constructed so far. We maintain the property that $T_i$ is a branching, that is, a collection of node-disjoint in-arborescences spanning all nodes. In an iteration, for each such in-arborescence we choose an arc of minimum weight with respect to $\price_C$ that goes from the root of the in-arborescence to some other component. We add these arcs to $T_i$, and for each directed cycle created, we delete one of its arcs. This results in a graph $T_{i+1}$ with at most half the number of weakly connected components that $T_i$ has, all being in-arborescences. We repeat this until the number of components becomes at most $\max\{1,m/k\}$. To reach this, we need at most $\lceil\log k\rceil$ iterations. Finally, we choose one of the roots of the components and add an arc from all the other roots to this one, obtaining a spanning in-arborescence $T$. 

It remains to show that $T$ also satisfies (ii). In the final stage, we add at most $\max\{1,m/k\}-1$ arcs to $T$, which corresponds to at most $k(\max\{1,m/k\}-1)\leq \max\{0,m-k\}$ clauses in $\Phi_T$. Now we bound the rest of $\Phi_T$. In iteration $i$, components of $T_i$ define a partition $\mathcal{B}=\mathcal{B}_1\cup\dots\cup\mathcal{B}_q$. Let us denote by $B_j$ the body corresponding to the root of the arborescence with node-set $\cB_j$. Let us consider the arcs $\{(B_j,B'_j)\mid j=1,\dots,q\}$ chosen to be added in the $i$th iteration. Now we obtain 
\begin{equation*}
|\Phi_{T_{i+1}\setminus T_i}|_C\leq \displaystyle\sum_{j=1}^q\price_C(B_j,B'_j)=
\displaystyle\sum_{j=1}^q\min_{B\in\cB\setminus\cB_j}\price_C(B_j,B)\leq OPT_C(\cB).
\end{equation*}
The first inequality follows from the construction of $T$. The equality follows from the criterion to choose the arcs to be added. The last inequality follows from Lemma~\ref{lem:main}. Since we have at most $\lceil\log k\rceil$ iterations, the lemma follows. \qed
\end{proof}

\begin{theorem} \label{thm:emin}
For key Horn functions, there exists a polynomial time \linebreak $\min\{\lceil\log n\rceil + 1,\lceil\log k\rceil + 2,k\}$-approximation algorithm for (C) and (BC), where $k$ is an upper bound on the sizes of bodies in $\cB$. 
\end{theorem}
\begin{proof}
We first show that $\Phi$ provided by Procedure~\ref{proc:base} is a $\min\{\lceil\log n\rceil + 1,\lceil\log k\rceil\allowbreak+2\}$-approximation for (C) and (BC). Note that $\Phi$ is a subformula of $\Psi_\cB$ defined by \eqref{eq:psi} since all bodies in $\Phi$ are from $\cB$. Furthermore, by our construction, $F_\Phi(B)=V$ for all $B\in\cB$. This implies that the output $\Phi$ represents $h_\cB$. Using Lemma~\ref{lem:minarb} and the fact that we added $|V\setminus B_0|\leq n$ clauses to $\Phi_T$ in Step 2, we obtain 
\begin{equation*}
|\Phi|_C\leq\lceil\log k\rceil OPT_C(\cB)+\max\{0,m-k\}+n.
\end{equation*}
By Lemma~\ref{lem:lb}, this gives a $(\lceil\log k\rceil + 2)$-approximation, while setting $k=n$ gives a $(\lceil\log n\rceil + 1)$-approximation. By Lemma~\ref{lem:origbod}, $OPT_{BC}(\cB)=|\cB|+OPT_{C}(\cB)$. Since $|\Phi|_{BC}=|\cB|+|\Phi|_{C}$, the same approximation ratios as above follow for (BC) as well.

Finally, Lemma~\ref{lem:kapprox} provides a different CNF that is a $k$-approximation for (C) and (BC). \qed
\end{proof}

\subsection{Literal minimum representations} \label{sec:litmin}

In this section we consider (L). The first difficulty that we have to overcome is that, unlike in the case of (C) and (BC), computing $\price_L$ is NP-hard as we show in Section~\ref{sec:pricenp}. To circumvent this, we give an $O(1)$-approximation algorithm for $\price_L(S,S')$ for any pair of sets $S,S'\subseteq V$. Note that if $S$ does not contain a body $B\in\cB$ then $\price_L(S,S')=\infty$, hence we assume that this is not the case.

We first analyze the structure of a pure Horn CNF $\Phi$ attaining the minimum in \eqref{eq:price} for (L). Starting the forward chaining procedure from $S$ with respect to $\Phi$, let $W_i$ denote the set of variables reached within the first $i$ steps. That is, $S=W_0\subsetneq W_1\subsetneq\dots\subsetneq W_t\supseteq S'$. We choose $\Phi$ in such a way that $t$ is as small as possible.
Let $B_i\in\mathcal{B}$ be a smallest body in $W_i$ for $i=0,\ldots,t-1$ and set $B_t:=S'$.

\begin{proposition} \label{cl:1}
$B_i\not\subseteq W_{i-1}$ for $i=1,\ldots,t$.
\end{proposition}
\begin{proof}
Suppose to the contrary that $B_i\subseteq W_{i-1}$ for some $1\leq i \leq t-1$. By the definition of forward chaining, every variable $v\in W_{i+1}\setminus W_{i}$ is reached through a clause $B\rightarrow v$ where $B\cap (W_i\setminus W_{i-1})\neq\emptyset$. Now substitute each such clause by $B_i\rightarrow v$. As $|B_i|\leq |B|$, the $|\cdot|_L$ size of the CNF does not increase. However, the number of steps in the forward chaining procedure decreases by at least one, contradicting the choice of $\Phi$. Finally, $S'=B_t\subseteq W_{t-1}$ would contradict the minimality of $t$.
\qed
\end{proof}

Proposition~\ref{cl:1} immediately implies that $|B_0|>|B_1|>\ldots>|B_{t-1}|$.

\begin{proposition} \label{cl:2}
$W_{i+1}\setminus W_i\subseteq B_{i+1}$ for $i=0,\ldots,t-1$.
\end{proposition}
\begin{proof}
Let $i$ be the smallest index that violates the condition. Take an arbitrary variable $v\in W_{i+1}\setminus W_i$. Then $v$ is reached in the $(i+1)$th step of the forward chaining procedure from a body of size at least $|B_i|$. If we substitute this clause by $B_{i+1}\rightarrow v$, the resulting CNF still satisfies $F_{\Phi}(B_0)\supseteq S'$ but has smaller $|\cdot|_L$ size by $|B_{i+1}|<|B_{i}|$, contradicting the minimality of $\Phi$. \qed
\end{proof}

By Proposition~\ref{cl:2}, $W_{i+1}\setminus W_i=B_{i+1}\setminus(S\cup \bigcup_{j=1}^{i} B_j)$. Define 
\[
\Phi^{(1)}:=\bigwedge_{i=0}^{t-1} B_i\rightarrow (B_{i+1}\setminus(S\cup\bigcup_{j=1}^{i} B_j)). 
\]
Observe that $\Phi^{(1)}$ has a simple structure which is based on a linear order of bodies $B_0,\dots,B_t$.

\begin{proposition} \label{cl:3}
$|\Phi^{(1)}|_L=|\Phi|_L$.
\end{proposition}
\begin{proof}
Take an arbitrary variable $v\in B_{i+1}\setminus(S\cup \bigcup_{j=1}^{i} B_j)$ for some $i=0,\ldots,t-1$. By the observation above, $v\in W_{i+1}\setminus W_i$. This means that $\Phi$ has at least one clause entering $v$, say $B\rightarrow v$, for which $B\subseteq W_i$ and so $|B|\geq |B_i|$. However, $\Phi^{(1)}$ has exactly one clause entering $v$, namely $B_i\rightarrow v$. This implies that $|\Phi^{(1)}|_L\leq |\Phi|_L$, and equality holds by the minimality of $\Phi$. \qed
\end{proof}

The proposition implies that $\Phi^{(1)}$ also realizes $\price_L(S,S')$. We know no efficient algorithms to compute $\Phi^{(1)}$, thus, using the next two propositions, we define a CNF that approximates $\Phi^{(1)}$ well and can be computed efficiently.

Let $i_0=0$ and for $j>0$ let $i_j$ denote the smallest index for which $|B_{i_j}|\leq |B_{i_{j-1}}|/2$. Let $r-1$ be the largest value for which $B_{i_{r-1}}$ exists and set $B_{i_r}:=S'$. Now define
\[
\Phi^{(2)}:=\bigwedge_{j=0}^{r-1} B_{i_j}\rightarrow (B_{i_{j+1}}\setminus(S\cup \bigcup_{\ell=1}^{j} B_{i_\ell})).
\]
It is easy to see that $F_{\Phi^{(2)}}(S)\supseteq S'$.

\begin{proposition} \label{cl:4}
$|\Phi^{(2)}|_L\leq 2|\Phi^{(1)}|_L$.
\end{proposition}
\begin{proof}
Take an arbitrary variable $v\in B_{i_{j+1}}\setminus(S\cup \bigcup_{\ell=1}^{j} B_{i_\ell})$ for some $j=0,\dots,r-1$. Then both $\Phi^{(1)}$ and $\Phi^{(2)}$ contain a single clause entering $v$. Namely, $v$ is reached from $B_{i_{j+1}-1}$ in $\Phi^{(1)}$ and from $B_{i_j}$ in $\Phi^{(2)}$. By the definition of the sequence $i_0,i_1,\ldots,i_{r-1}$, we get $|B_{i_j}|\leq 2|B_{i_{j+1}-1}|$, concluding the proof. \qed
\end{proof}

Although $\Phi^{(2)}$ gives a $2$-approximation for $|\Phi|_L$, it is not clear how we could find such a representation. Define 
\[
\Phi^{(3)}:=\bigwedge_{j=0}^{r-1} B_{i_j}\rightarrow (B_{i_{j+1}}\setminus (S\cup B_{i_j})). 
\]
The only difference between $\Phi^{(2)}$ and $\Phi^{(3)}$ is that we add unnecessary clauses to the representation. However, the next claim shows that the size of the formula cannot increase a lot.

\begin{proposition} \label{cl:5}
$|\Phi^{(3)}|_L\leq \frac{27}{17}|\Phi^{(2)}|_L$.
\end{proposition}
\begin{proof}
Take an arbitrary variable $v$ that appears as the head of a clause in the representation $\Phi^{(3)}$. Let $j$ be the smallest index for which $v\in B_{i_{j+1}}\setminus(S\cup \bigcup_{\ell=1}^{j} B_{i_\ell})$. Then $\Phi^{(2)}$ contains a single clause entering $v$, namely $B_{i_j}\rightarrow v$. On the other hand, the set $\{B_{i_j}\rightarrow v\}\cup\{B_{i_\ell}\rightarrow v\mid \ell=j+2,\dots,r-1\}$ contains all the clauses of $\Phi^{(3)}$ that enter $v$. By the definition of the sequence $i_0,i_1,\ldots,i_{r-1}$, we get $\sum_{\ell=j+2}^{r-1}(|B_{i_\ell}|+1)= (r-j-2)+\sum_{\ell=j+2}^{r-1}|B_{i_\ell}|\leq \lfloor\log|B_{i_{j+1}}|\rfloor+|B_{i_j}|/2-1\leq \lfloor\log|B_{i_{j}}|\rfloor+|B_{i_j}|/2-2$. We get at most this many extra literals in $\Phi^{(3)}$ on top of the $|B_{i_j}|+1$ literals in $\Phi^{(2)}$. As $\lfloor\log x\rfloor/(x+1)+x/(2(x+1))-2/(x+1)\leq 10/17$ for $x\in\mathbb{Z}_+$, the statement follows. \qed
\end{proof}

By Propositions~\ref{cl:3},~\ref{cl:4} and~\ref{cl:5},
\begin{equation}
|\Phi^{(3)}|_L\leq \frac{27}{17}|\Phi^{(2)}|_L\leq \frac{54}{17}|\Phi^{(1)}|_L=\frac{54}{17}|\Phi|_L. \label{eq:approx}
\end{equation}

\begin{lemma}\label{lem:lmin}
There exists an efficient algorithm to construct a CNF $\Lambda(S,S')$ such that $|\Lambda(S,S')|_L\leq \frac{54}{17}\price_L(S,S')$, $\cB_{\Lambda(S,S')}\subseteq\cB$, and $F_{\Lambda(S,S')}(S)\supseteq S'$.
\end{lemma}
\begin{proof}
We consider an extension of the body graph by adding $S'$ to $V(D_\cB)$. We also define arc-weights by setting $w(B,B'):=|B'\setminus (S\cup B)|(|B|+1)$ for $B,B'\in\mathcal{B}\cup\{S'\}$. Let $B_0$ be a smallest body contained in $S$ (as defined before Proposition~\ref{cl:1}). Compute a shortest path $P$ from $B_0$ to $S'$ and define
\begin{equation}
\Lambda(S,S') = \bigwedge_{(B,B')\in P} B\rightarrow(B'\setminus(S\cup B)). \label{eq:path}
\end{equation}
Note that, by definition, $|\Lambda(S,S')|_L$ is the weight of the shortest path $P$, while $|\Phi^{(3)}|_L$ is the length of one of the paths from $S$ to $S'$. By \eqref{eq:approx}, $|\Lambda(S,S')|_L\leq |\Phi^{(3)}|_L\leq \frac{54}{17}|\Phi|_L$. That is, $\Lambda(S,S')$ provides a $\frac{54}{17}$-approximation for $\price_L(S,S')$ as required, finishing the proof of the lemma. \qed
\end{proof}

We prove that the algorithm described in Procedure 2 provides the stated approximated factor for (L). We note that a minimum weight spanning in-arborescence of a directed graph rooted at a fixed node can be found in polynomial time, see \cite{CL65,E67}. Let $B_{\min}$ be a smallest body in $\cB$ and denote $\mathcal{B}'=\cB\setminus\{B_{\min}\}$. We define the weight of an arc $(B,B')$ in the body graph to be $w(B,B')=|\Lambda(B,B')|_L$.

\begin{algorithm2e}
\caption{Approximation of (L)} \label{proc:litmin}
    
 Let $B_{\min}$ be a smallest body in $\cB$.\    

 Set $w(B,B')=|\Lambda(B,B')|_L$ for $(B,B')\in E(D_{\mathcal{B}})$.\

 Determine a minimum $w$-weight spanning in-arborescence $T$ of $D_{\mathcal{B}}$ such that $T$ is rooted at $B_{\min}$. 
 
 Output $\Phi=\bigwedge_{(B,B')\in T}\Lambda(B,B')\wedge (B_{\min}\rightarrow (V\setminus B_{\min}))$. \linebreak /$*$ Here $\Lambda(B,B')$ is defined as in \eqref{eq:path}. $*$/
\end{algorithm2e}

Choose a smallest body $B_{\min}$ in $\cB$ and let $\delta:=|B_{\min}|$. Set $w(B,B'):=|\Lambda(B,B')|_L$ for $(B,B')\in E(D_{\mathcal{B}})$.

\begin{lemma} \label{lem:minarb2}
Let $T$ denote a minimum $w$-weight spanning in-arborescence in $D_\cB$ such that $T$ is rooted at $B_{\min}$. Then $$\left|\bigwedge_{(B,B')\in T}\Lambda(B,B')\right|_L\leq\left(\frac{108}{17}\lceil\log k\rceil+1\right) OPT_L(\cB),$$ where $k$ is the size of a largest body in $\cB$.
\end{lemma}
\begin{proof}
We construct a subgraph $T$ of $D_\cB$ such that (i) it is a spanning in-arborescence, and (ii) $|\bigwedge_{(B,B')\in T}\Lambda(B,B')|_L\leq (2\lceil\log k\rceil +1) OPT_L(\cB)$. We start with the directed graph $T_1$ on node set $\cB$ that has no arcs. In a general step of the algorithm, $T_i$ will denote the graph constructed so far. We maintain the property that $T_i$ is a branching, that is, a collection of node-disjoint in-arborescences spanning all nodes. In an iteration, for each such in-arborescence we choose an arc of minimum weight with respect to $w$ that goes from the root of the in-arborescence to some other component. We add these arcs to $T_i$, and for each directed cycle created, we delete one of its arcs. This results in a graph $T_{i+1}$ with at most half the number of weakly connected components that $T_i$ has, all being in-arborescences. We repeat this until the number of components becomes at most $\max\{1,m/k^2\}$. To reach this, we need at most $\lceil\log k^2\rceil\leq 2\lceil\log k\rceil$ iterations. Finally, we add an arc from all the other roots to $B_{\min}$ and delete all the arcs leaving $B_{\min}$, obtaining a spanning in-arborescence $T$ rooted at $B_{\min}$. 

It remains to show that $T$ also satisfies (ii). In the final stage, we add at most $\max\{1,m/k^2\}$ arcs to $T$ whose total weight is upper bounded by $(k+1)\delta\max\{1,m/k^2\}\leq \max\{n\delta,2m\}\leq OPT_L(\cB)$, where the last inequality follows by Lemma \ref{lem:lb}.  Now we bound the rest of $\bigwedge_{(B,B')\in T}\Lambda(B,B')$. In iteration $i$, components of $T_i$ define a partition $\mathcal{B}=\mathcal{B}_1\cup\dots\cup\mathcal{B}_q$. Let us denote by $B_j$ the body corresponding to the root of the arborescence with node-set $\cB_j$. Let us consider the arcs $\{(B_j,B'_j)\mid j=1,\dots,q\}$ chosen to be added in the $i$th iteration. Now we obtain 
\begin{eqnarray*}
\left|\displaystyle\bigwedge_{(B,B')\in T_{i+1}\setminus T_i}\Lambda(B,B')\right|_L&=& \displaystyle\sum_{j=1}^q w(B_j,B'_j)=
\displaystyle\sum_{j=1}^q\min_{B\in\cB\setminus\cB_j}w(B_j,B)\\
&\leq&\frac{54}{17}\displaystyle\sum_{j=1}^q\min_{B\in\cB\setminus\cB_j}\price_L(B_j,B)\leq \frac{54}{17}OPT_L(\cB),
\end{eqnarray*}
where the first and second inequalities follow by Lemmas \ref{lem:lmin} and \ref{lem:main}, respectively. 
Since we have at most $2\lceil\log k\rceil$ iterations, the lemma follows. \qed
\end{proof}

\begin{theorem} \label{thm:lmin}
For key Horn functions, there exists a polynomial time \linebreak $\min\{\frac{108}{17}\lceil\log k\rceil+2,k\}$-approximation algorithm for (L), where $k$ is the size of a largest body in $\cB$. 
\end{theorem}
\begin{proof}
We first show that $\Phi$ provided by Procedure~\ref{proc:litmin} is a $(\frac{108}{17}\lceil\log k\rceil+2)$-approximation for (L). Note that $\Phi$ is a subformula of $\Psi_\cB$ defined by \eqref{eq:psi} since all bodies in $\Phi$ are from $\cB$. Furthermore, by our construction, $F_\Phi(B)=V$ for all $B\in\cB$. This implies that the output $\Phi$ represents $h_\cB$. By Lemma~\ref{lem:lb}, we add at most $n(\delta+1)\leq OPT_L(\cB)$ literals to $\bigwedge_{(B,B')\in T}\Lambda(B,B')$ in Step 4. This, together with Lemma~\ref{lem:minarb2}, implies the theorem.\qed
\end{proof}

\section{Hardness of computing $\price_L$} \label{sec:pricenp}

In this section we prove that computing $\price_L$ is NP-hard. Let $S$ be a ground set. Given a sequence $\cS=(S_0,S_1,...,S_s)$ of subsets of $S$, we associate to it a CNF
\begin{equation}
\Phi_\cS = \bigwedge_{i=0}^{s-1} \left(S_i\to \left(S_{i+1}\setminus\bigcup_{j\leq i}S_j\right)\right). \label{eq:phis}
\end{equation}
We denote by $\cost_L(\cS)=\cost_L(S_0,...,S_s)$ the $L$-measure (number of literals) of $\Phi_\cS$, i.e.,
\begin{equation*}
\cost_L(\cS)=\cost_L(S_0,...,S_s)=\sum_{i=0}^{s-1} \left(|S_i|+1\right)\cdot \left|S_{i+1}\setminus\left(\bigcup_{j\leq i} S_j\right)\right|.
\end{equation*}
Let us note that we use $\cS$ both as a family and a sequence of subsets. This is because in this section we are concerned with shortest sequences between given sets $S_0$ and $S_s$ that minimizes $\cost_L(\cS)$ and by Proposition \ref{cl:1} we can assume for such sequences that $|S_0|>|S_1|>\dots >|S_{s-1}|$.

The following simple lemma is central to our construction.

\begin{lemma}\label{l0}
For three sets $A$, $B$, and $C$ assume $E=B\setminus(A\cup C)$,  $F=B\cap (C\setminus A)$, $G=C\setminus (A\cup B)$. Furthermore assume that $|A|=a$, $|B|=b$, $|C|=c$, $|E|=e$, $|F|=f$ and $|G|=g$. Then the followings are equivalent.
\begin{enumerate}[label=(\alph*)]
\item $\cost_L(A,B,C) < \cost_L(A,C)$, \label{it:i}
\item $(a-b)\cdot g ~>~ (a+1)\cdot e$,
\item$a\cdot (g-e) ~>~ e+ b\cdot g$.
\end{enumerate}
\end{lemma}
\begin{proof}
The claim follows by elementary computations using the expressions
\begin{align*}
\cost_L(A,B,C) & = (a+1)\cdot (e+f)+(b+1)\cdot g,  \qquad\text{ and }\\ 
\cost_L(A,C) & = (a+1)\cdot (f+g).
\end{align*}
\qed
\end{proof}

Consider a $3$-CNF (exactly $3$ literals in each clause) $\Phi=\bigwedge_{k=1}^m C^0_k$ in which every variable $x_i$, $i=1,...,n$ appears at most 4 times. SAT is NP-complete for this family of CNFs \cite{tovey1984simplified}. Let us complement the literals in the clauses in all possible ways, and denote by $C^j_k$, $j=1,...,7$, $k=1,...,m$ the clauses we obtain in this way from the ones appearing in $\Phi$. Let $M=\{C^j_k\mid j=0,...,7, ~ k=1,...,m\}$, and by abuse of notation view $\Phi$ as a subset of $M$. Note that for all $i$, both variable $x_i$ and its complement $\bar{x}_i$ appear at most $\delta_i\leq 16$ times in the clauses of $M$.

Define sets $T$, $B_j$, $j=0,...,n$ and $A_j$, $j=1,...,n+1$ to be pairwise disjoint and disjoint from $M$. Denote $|T|=\tau$, $|A_j|=\alpha$ for $j=1,...,n+1$, and $|B_j|=\beta$, $j=0,...,n$.

We define 
\begin{align*}
X_i &= \displaystyle \left(\bigcup_{j=i}^n B_j\right) \cup \left(\bigcup_{j=1}^i A_j\right) \cup \left\{ C^j_k\in M \mid x_i\in C^j_k\right\}, \qquad\text{ and}\\
Y_i &= \displaystyle \left(\bigcup_{j=i}^n B_j\right) \cup \left(\bigcup_{j=1}^i A_j\right) \cup \left\{ C^j_k\in M \mid \bar{x}_i\in C^j_k\right\},
\end{align*}
for $i=0,...,n+1$. Note that since $x_0$ and $x_{n+1}$ are not variables of $\Phi$, we have $X_0=Y_0=B_0\cup\dots \cup B_n$ and $X_{n+1}=Y_{n+1}=A_1\cup\dots \cup A_{n+1}$. Furthermore, let us define $S=X_0$, $Z=X_{n+1}\cup \Phi$, and set
\begin{equation}
\cB_\Phi = \{S,Z,T\} \cup \{X_i,Y_i\mid i=1,...,n\}. \label{hypg}
\end{equation}

\smallskip

Our plan is to choose $\tau \gg \beta \gg \alpha \gg \max\{n,m\}$ such that we have
\[
|S| > |X_1|=|Y_1| > \dots > |X_n|=|Y_n| > |Z|.
\]
Given this, let us recall that an optimal solution realizing $price_L(S,T)$ with respect to the family $\cB_\Phi$ involve sets from $\cB_\Phi$ in strictly decreasing order of their size. In what follows, we show first that, with a right choice of parameters, such an optimal solution must include $Z$, and must include exactly one of $X_i$ and $Y_i$ for all $i=1,\dots ,n$. 

\smallskip

Define further $\delta_0=\delta_{n+1}=0$ and $\delta_i=|X_i\cap M|$ for $i=1,...,n$.  With these notation, we have the following easy to see relations that we will rely on in the proof without mentioning them explicitly:
\begin{enumerate}[label=(\roman*)]
  \item $\delta_i=|X_i\cap M|=|Y_i\cap M|\leq 16$ for $i=0,...,n+1$, \label{prop-i}
  \item $|S|=(n+1)\beta$, $|Z|=(n+1)\alpha+m$,
  \item $|X_i|=|Y_i|=(n-i+1)\beta+i\alpha+\delta_i$ for $i=0,...,n+1$, \label{prop-ii}
  \item $|X_i| > |X_{i+1}|+\alpha$ for $i=0,...,n$, \label{prop-iii}
  \item $\alpha ~\leq~ |X_i\setminus \left(\bigcup_{j=0}^{i-1}X_j\right)| ~\leq~ \alpha+16$ for $i=1,...,n+1$, \label{prop-iv}
  \item $X_j\cap (X_{i+1}\setminus X_i)\subseteq X_{i+1}\cap M$ for $i=1,...,n$ and $j<i$. \label{prop-v}
\end{enumerate}

Note that for \eqref{prop-iii} to hold it is enough to have 
\begin{equation}\label{e-1st}
\beta ~>~ 2\alpha + 16.
\end{equation}
\smallskip

For $\sigma\in\{0,1,*\}^{[n]}$, where $[n]=\{1,\dots,n\}$, let us define $\cP(\sigma)$ as the sequence of sets from $\cB_\Phi\setminus\{S,Z,T\}$ such that $\cP(\sigma)$ contains $X_i$ iff $\sigma_i=1$ and it contains $Y_i$ iff $\sigma_i=0$, for all $i=1,...,n$. Furthermore, for $\xi\in\{0,1\}$ we use the notation
\begin{equation*}
X_i^\xi ~=~ 
\begin{cases}
X_i & \text{ if }~~\xi=1,\\
Y_i & \text{ if }~~\xi=0.
\end{cases}
\end{equation*}

\begin{lemma} \label{lem:z}
For all $\sigma\in\{0,1,*\}^{[n]}$, we have
\begin{equation*}
\cost_L(S,\cP(\sigma),T) > \cost_L(S,\cP(\sigma),Z,T)
\end{equation*}
whenever
\begin{equation}\label{e-2nd}
(\beta-\alpha-m)\cdot\tau ~>~ ((n+1)\beta + 17)\cdot((n+1)\alpha +m).
\end{equation}
\end{lemma}
\begin{proof}
Let $\sigma_0=1$ and define $X^{\sigma_0}_0=S=X_0=Y_0$. Assume that $i$ is the largest index such that $\sigma_i\in\{0,1\}$. 
Since $S=X_0=Y_0$, such an $i$ exists and $0\leq i \leq n$. We show that $\cost(X_i,Z,T)<\cost(X_i,T)$, thus proving the lemma.

Let us apply Lemma \ref{l0} with $A=X_i^{\sigma_i}$, $B=Z$ and $C=T$. We have 
$a=|X_i|=|Y_i|=(n-i+1)\beta + i\alpha +\delta_i$, $b=|Z|=(n+1)\alpha +m$, $c=|T|=\tau$, $e\leq (n-i+1)\alpha +m$, $f=0$, and $g=\tau$. It is enough to show that $(a-b)\cdot g>(a+1)\cdot e$, that is, it suffices to verify 
\begin{eqnarray*}
&\left((n-i+1)\beta+(i-n-1)\alpha+\delta_i-m\right)\tau >&\\
&\left((n-i+1)\beta + i\alpha +\delta_i+1\right)\left((n-i+1)\alpha +m\right),&
\end{eqnarray*}
which follows by \eqref{e-2nd}.
\qed
\end{proof}

For  $\sigma\in\{0,1,*\}^{[n]}$  if $\sigma_j=*$, then let us denote by $\sigma^{j\to 0}$ 
and $\sigma^{j\to 1}$ the sequences obtained by switching the $j$th entry in $\sigma$ to $0$ and $1$, respectively.

\begin{lemma} \label{lem:insert}
For every $\sigma\in\{0,1,*\}^{[n]}$ with $\sigma_j=*$, we have
\begin{equation*}
\cost_L(S,\cP(\sigma),Z,T) > \cost_L(S,\cP(\sigma^{j\to \epsilon}),Z,T), 
\end{equation*}
for all $\epsilon\in\{0,1\}$, whenever
\begin{equation}\label{e-3rd}
(\beta-\alpha-16)\cdot\alpha ~>~ 16\cdot((n+1)\beta + 17)
\end{equation}
\end{lemma}
\begin{proof}
Let $\sigma_0=\sigma_{n+1}=1$ and define $X^{\sigma_0}_0=S$ and $X^{\sigma_{n+1}}_{n+1}=Z$. Choose an arbitrary index $1\leq j\leq n$ with $\sigma_j=*$, and set $i$ to be the largest index $i<j$ with $\sigma_i\neq*$ while $k$ to be the smallest index $j<k$ with $\sigma_i\neq *$. As $\sigma_0=\sigma_{n+1}=1$, such $i$ and $k$ exist.

We apply Lemma \ref{l0} with $A=X_i^{\sigma_i}$, $B=X_j^{\sigma_j}$ and $C=X_k^{\sigma_k}$. We have $a=(n-i+1)\beta+i\alpha+\delta_i$, $b=(n-j+1)\beta+j\alpha+\delta_j$, $g\geq (k-j)\alpha$ and $e\leq\delta_{j}$. In order $\ref{it:i}$ to be true, we need $(a-b)\cdot g > (a+1)\cdot e$, hence it suffices to show that
\[
\displaystyle \left((j-i)\beta+(i-j)\alpha+\delta_i-\delta_j\right)\left(k-j\right)\alpha > \left((n-i+1)\beta+i\alpha+\delta_i+1\right)\delta_j,    
\]
which follows by \eqref{e-3rd}.
\qed
\end{proof}

\smallskip
Let us note that \eqref{e-1st}, \eqref{e-2nd}, and \eqref{e-3rd} will hold if
\begin{align}
\alpha^2 &> \max\{m^2, 16^2\cdot(n+1)+ 2\cdot 16^2\cdot(n+1)^2+16\cdot 17\}\label{e-1}\\
\beta &> 2\alpha +32(n+1) + 16\label{e-2}\\
\tau &> \left((n+1)\beta + 17\right)\cdot \left((n+1)\alpha + m\right)\label{e-3}
\end{align}

\smallskip

It is easy to see that we can choose $\alpha$, $\beta$, and $\tau$ such that \eqref{e-1}, \eqref{e-2}, and \eqref{e-3} hold, and none of these parameters exceed $m^2n^3$, thus our construction above has polynomial size in the the size of $\Phi$. Let us assume for the rest of our proof that we choose these parameters satisfying \eqref{e-1}, \eqref{e-2}, and \eqref{e-3}, and as small as possible. 

In what follows we show that $\price_L(S,T)$ is the smallest if and only if $\Phi$ is satisfiable. 

\smallskip

For an index $i\in [n]$ and $\sigma\in\{0,1\}^{[n]}$ let us define 
\begin{equation*}
W_i(\sigma)=S\cup \bigcup_{j=1}^i X_j^{\sigma_j}.
\end{equation*}
Furthermore, define $W_0(\sigma)=S$.

\begin{lemma}\label{l-4}
There exists a function $d:[n]\to \mathbb{Z}_+$ such that 
\begin{equation*}
|X_{i+1}\setminus W_i(\sigma)| = |Y_{i+1}\setminus W_i(\sigma)| = d(i).
\end{equation*}
for every $i=0,\dots,n$ and $\sigma\in\{0,1\}^{[n]}$.
\end{lemma}
\begin{proof}
To see the claim, let us consider a clause $C$ of $\Phi$ that contains variable $x_{i+1}$ or its negation. Let us denote by $\cC(C)\subseteq M$ the set of eight clauses included in $M$, obtained from $C$ by complementing the three literals of $C$ in all possible ways. Let us further denote by $I(C)$ the indices of the variables that are involved (with or without a complementation) in $C$. Let us then observe that if $i+1$ is the smallest index in $I(C)$, 
then  both $X_{i+1}\setminus W_i(\sigma)$ and $Y_{i+1}\setminus W_i(\sigma)$ contain exactly $4$ elements of $\cC(C)$; if $i+1$ is the second smallest index in $I(C)$, 
then  both $X_{i+1}\setminus W_i(\sigma)$ and $Y_{i+1}\setminus W_i(\sigma)$ contain exactly $2$ elements of $\cC(C)$; while if $i+1$ is the largest index in $I(C)$, 
then  both $X_{i+1}\setminus W_i(\sigma)$ and $Y_{i+1}\setminus W_i(\sigma)$ contain exactly $1$ element of $\cC(C)$. Note that these numbers do not depend on $\sigma\in\{0,1\}^{[n]}$, and hence the claim follows. 
\qed
\end{proof}

\begin{lemma}\label{l-5}
There exists an integer $g\in\mathbb{Z}_+$ such that
\begin{equation*}
\cost_L(S,\cP(\sigma)) = g
\end{equation*}
for every $\sigma\in\{0,1\}^[n]$. 
\end{lemma}
\begin{proof}
The claim follows by Lemma \ref{l-4} and the fact that $|X_i|=|Y_i|$ for $i=1,\dots,n$.
\end{proof}

\begin{lemma} \label{lem:C}
There exists an integer $C$ such that for all $\sigma\in\{0,1\}^{[n]}$ we have
\[
\cost_L(S,\cP(\sigma),X_{n+1}) = C.
\]
\end{lemma}
\begin{proof}
From Lemma~\ref{l-5}, we get $C=g+\alpha(|X_n|+1)$ and the statement follows.
\qed
\end{proof}

\begin{lemma} \label{lem:final}
For  $\sigma\in\{0,1\}^{[n]}$ we have
\[
\cost_L(S,\cP(\sigma),Z,T) = C + |X_n|\cdot |\Phi(\sigma)| + |Z|\cdot |T|,
\]
where $|\Phi(\sigma)|$ denotes the number of clauses of $\Phi$ that are not satisfied by $\sigma$.
\end{lemma}
\begin{proof}
The lemma follows by the construction and by Lemma~\ref{lem:C}.
\qed
\end{proof}

\begin{lemma} \label{lem:opt}
For the hypergraph $\cB_\Phi$ defined in \eqref{hypg} we have
\[
\price_L(S,T) = C+|Z|\cdot |T|
\]
if and only if $\Phi$ is satisfiable.
\end{lemma}
\begin{proof}
The construction of $\Phi^{(1)}$ in Section~\ref{sec:litmin} shows that there exists a pure Horn CNF attaining the minimum in $\price_L(S,T)$ that can be written in form \eqref{eq:phis} for some sequence $\{S_0,\dots,S_s\}\subseteq \mathcal{B}_\Phi$ where $|S_0|> |S_1|>...>|S_s|$. By Lemmas~\ref{lem:z} and \ref{lem:insert}, we may assume that $\mathcal{S}=\{S,\mathcal{P}(\sigma),Z,T\}$ for some truth assignment $\sigma\in\{0,1\}^[n]$. Lemma~\ref{lem:final} implies that $\price_L(S,T)=\cost_L(S,\cP(\sigma),Z,T)=C +|Z|\cdot |T|$ if and only if $|\Phi(\sigma)|=0$, that is, if $\sigma$ is a true point of $\Phi$. 
\qed
\end{proof}

\begin{theorem}
Computing $\price_L$ is NP-hard.
\end{theorem}
\begin{proof}
Let $\Phi$ be a $3$-CNF in which every variable appears at most $4$ times. Recall that SAT is NP-complete even when restricted to this class of CNF formulas \cite{tovey1984simplified}. By Lemma~\ref{lem:opt}, $\Phi$ is satisfiable if and only if $\price_L(S,T) = C+|Z|\cdot |T|$ that is if and only if there exists a $\sigma\in\{0,1\}^{[n]}$  such that $|\Phi(\sigma)|=0$. This shows that computing $\price_L$ is NP-hard. 
\qed
\end{proof}

\section{Clause minimization and minimum weight strongly connected subgraphs} \label{sec:ex}

Given a strongly connected graph $D=(V,E)$ and non-negative weights $w:E\rightarrow \mathbb{Z}_+$, we denote by $\mathtt{MWSCS}(D,w)$ the problem of finding a minimum weight subset $F\subseteq E$ of the arcs such that $(V,F)$ is also strongly connected. We denote by $\mathtt{mwscs}(D,w)=w(F)$ the weight of such a minimum weight arc subset. $\mathtt{MWSCS}$ is an NP-hard problem, for which polynomial time approximation algorithms are known. For the case of uniform weights a $1.61$-approximation was given by Khuller et al. \cite{Khuller}. For general weights a simple $2$-approximation is due to Fredericson and J\'aj\'a \cite{jaja}. Note that in the case of general weights, we can assume that $D$ is a complete directed graph.

As it was observed already in the beginning of  Section~\ref{sec:approx}, there is a natural relation of the above problem to the minimization of a key Horn function. Let us consider a Sperner hypergraph $\mathcal{B}\subseteq 2^{V}\setminus\{V\}$ and the corresponding Horn function
\begin{equation}
h_\mathcal{B} ~=~ \bigwedge_{B\in\mathcal{B}} B\to (V\setminus B).
\end{equation} 
The body graph of $\mathcal{B}$ was a complete directed graph $D_\mathcal{B}$ where $V(D_\mathcal{B})=\mathcal{B}$. Define a weight function $w$ on the arcs of this graph by setting $w(B,B')=\price_*(B,B')$ for all $B,B'\in \mathcal{B}$, $B\neq B'$, where $\price_*$ is defined in \eqref{eq:price}. Then any solution $F\subseteq E(G_\mathcal{B})=\mathcal{B}\times \mathcal{B}$ of problem $\mathtt{MWSCS}(D_\mathcal{B},w)$ defines a representation of $h_\mathcal{B}$:
\begin{equation}
\Phi(F) ~=~ \bigwedge_{(B,B')\in E(G_\mathcal{B})} \Phi_*(B,B'),
\end{equation}
where $\Phi_*(B,B')$ is a formula for which $B'\subseteq F_{\Phi_*(B,B')}(B)$, $\mathcal{B}_{\Phi_*(B,B')}\subseteq \cB$ and $|\Phi_*(B,B')|_*=\price_*(B,B')$.
It is immediate to see that $OPT_*(\mathcal{B})\leq w(F)$ holds. Thus, it is natural to expect that a polynomial time approximation of problem $\mathtt{MWSCS}(D_\mathcal{B},w)$ provides also a good approximation for $OPT_*(\mathcal{B})$. This however turns out to be false for the case of $*=C$. 

Let us recall first some basic facts on finite projective spaces from the book \cite{Dembowski}.
The finite projective space $PG(d,q)$ of dimension $d$ over a finite field $GF(q)$ of order $q$ (prime power) has $n=q^d+q^{d-1}+\cdots +q+1$ points. Subspaces of dimension $k$ are isomorphic to $PG(k,q)$ for $0\leq k<d$, where $0$-dimension subspaces are the points themselves. The number of subspaces of dimension $k<d$ is 
\[
N_k(d,q) ~=~ \prod_{i=0}^{k} \frac{q^{d+1-i}-1}{q^{i+1}-1},
\]
and the number of points of such a subspace is $q^k+q^{k-1}+\cdots +q+1$. In particular, the number of subspaces of dimension $d-1$ is $N_{d-1}(d,q)=n$. If $F$ and $F'$ are two distinct subspaces of dimension $k$, then 
\[
2k-d \leq dim(F\cap F') \leq k-1.
\]
Furthermore, any $k+1$ points belong to at least one subspace of dimension $k$. 

Let us also recall that $PG(d,q)$ has a cyclic automorphism. In other words the points of $PG(d,q)$ can be identified with the integers of the cyclic group $\mathbb{Z}_n$ of modulo $n$ addition such that if $F\subseteq \mathbb{Z}_n$ is a subspace of dimension $k$, then $F+i=\{f+i \mod n\mid f\in F\}$ is also a subspace of dimension $k$ and $F$ and $F+i$ are distinct. Furthermore, if $X\subseteq \mathbb{Z}_n$ is a subspace of dimension $d-1$ then the family 
$\mathcal{X}=\{X+i\mid i\in\mathbb{Z}_n\}$ contains all subspaces of $PG(d,q)$ of dimension $d-1$. In the rest of this section we use $+$ for the modulo $n$ addition of integers.

\begin{lemma}\label{lfs0}
For every $k=0,...,d-1$ there exists a unique subspace of dimension $k$ that contains $\{0,1,...,k\}$. 
\end{lemma}
\begin{proof} 
By the properties we recalled above it follows that there is at least one such subspace for every $0\leq k<d$. 
We prove that there is at most one by induction on $k$.
For $k=0$ this is obvious, since the points are the only subspaces of dimension $0$. Assume next that the claim is already proved for all $k'<k$, and assume that there are two distinct subspaces, $F$ and $F'$, of dimension $k$ both of which contains the set $\{0,1,...,k\}$. Then $F\cap F'$ and $(F-1)\cap (F'-1)=(F\cap F')-1$ are two distinct subspaces of dimension $k'<k$ and both contain $\{0,1,...,k-1\}$, contradicting our assumption, and thus proving our claim.
\qed
\end{proof}

Thus, by Lemma \ref{lfs0} there exists a unique subspace $X\subseteq \mathbb{Z}_n$ of dimension $d-1$ that contains $\{0,1,...,d-1\}$. Let us also introduce the set $D=\{0,1,...,d\}$. 

\begin{lemma}\label{lfs1}
$d\not\in X$.
\end{lemma}
\begin{proof}
Assume to the contrary that $d\in X$. Then the set $\{0,1,...,d-1\}$ is contained by both $X$ and $X-1=X+(n-2)$, contradicting Lemma \ref{lfs0}, since $X$ and $X-1$ are distinct subspaces of dimension $d-1$. \qed
\end{proof}

\begin{theorem} \label{t-lowerbound}
Let $q$ be a prime power, $d$ be a positive integer, $n$ be the number of points of $PG(d,q)$, and $V=\mathbb{Z}_n$. Then we have
\begin{equation}
\max_{\cB\subseteq 2^V\setminus\{V\}} \frac{\mathtt{mwscs}(D_\mathcal{B},\price_C)}{OPT_C(\mathcal{B})} \geq \frac{n}{12}.
\end{equation}
\end{theorem}
\begin{proof}
Let us now define $\mathcal{B}:=\mathcal{X}\cup \{D+i\mid i\in\mathbb{Z}_n\}$, and observe that for any distinct pair $B\in\mathcal{X}$ and $B'\in\mathcal{B}$ we have $|B\setminus B'|\geq q^{d-1}$. Since in any solution $F\subseteq \mathcal{B}\times \mathcal{B}$ we must have an arc entering $B$ for all $B\in\mathcal{X}$, we get
\begin{equation}
\mathtt{mwscs}(D_\mathcal{B},\price_C) ~\geq~ n\cdot q^{d-1}.
\end{equation}

On the other hand, we have that
\begin{equation}
\Phi ~=~ \left(D\to (\mathbb{Z}_n\setminus D)\right)\wedge \left(\bigwedge_{i\in\mathbb{Z}_n} (X+i)\to d+i\right)\wedge \left(\bigwedge_{i\in\mathbb{Z}_n} (D+i)\to d+1+i\right)
\end{equation}
is a representation of $h_\mathcal{B}$ and $|\Phi|_C\leq 3n$. Choosing $q=2$ and $d>1$, we get 
\begin{equation}
\mathtt{mwscs}(G_\mathcal{B},\price_C) ~\geq~ \frac{n}{12}\cdot OPT_C(\mathcal{B}),
\end{equation}
completing the proof of the theorem.
\qed
\end{proof}

\section{Conclusions}

In this paper we study the class of key Horn functions which is a generalization of a well-studied class of hydra functions \cite{Sloan2017HydrasDH,Kucera2014HydrasCO}. Given a CNF representing a key Horn function, we are interested in finding the minimum size logically equivalent CNF,  where the size of the output CNF is measured in several different ways. This problem is known to be NP-hard already for hydra CNFs for most common measures of the CNF size.

The main results of the paper are two approximation algorithms for key Horn CNFs – one for minimizing the number of clauses and the other for minimizing the total number of literals in the output CNF. Both algorithms achieve a logarithmic approximation bound with respect to the size of the largest body in the input CNF (denoted by $k$). This parameter can be also defined as the size of the largest clause in the input CNF minus one. Note that $k$ is a trivial lower bound on the number of variables (denoted by $n$).

These algorithms are (to the best of our knowledge) first approximation algorithms for NP-hard Horn minimization problems that guarantee a sublinear approximation bound with respect to $k$. It follows, that both algorithms also guarantee a sublinear approximation bound with respect to $n$. There are two approximation algorithms for Horn minimization known in the literature, one for general Horn CNFs \cite{HK93}, and one for hydra CNFs \cite{Sloan2017HydrasDH}), but both of them guarantee only a linear (or higher) approximation bound with respect to $k$ (see Table~\ref{tab:results} and the relevant text in the introduction section for details).

For a given pair of sets $S,T$ and set of bodies $\cB$, we prove NP-hardness of the problem of finding a literal minimum CNF $\Phi$ that uses bodies only from $\cB$ and for which the forward chaining procedure starting from $S$ reaches all the variables in $T$. 

In opposed to our approach which takes an in-branching in the body graph and extends it with a small number of additional edges, we show that a polynomial time approximation of the minimum weight strongly connected subgraph problem in the body graph does not necessarily provides a good solution for the edge-minimum representation problem. The counterexample is based on a cunstruction using finite projective spaces.

\section*{Acknowledgement}
The research was supported by the J\'anos Bolyai Research Fellowship of the Hungarian Academy of Sciences, by Czech Science Foundation (Grant 19-19463S), and by SVV project number 260 453.

%
%
%
\bibliographystyle{splncs04}
\bibliography{keybib}
\end{document}